\documentclass[sn-vancouver,Numbered]{sn-jnl}

\usepackage{graphicx}
\usepackage{multirow}%
\usepackage{amsmath,amssymb,amsfonts}%
\usepackage{amsthm}%
\usepackage[title]{appendix}%
\usepackage{xcolor}%
\usepackage{textcomp}%
\usepackage{manyfoot}%
\usepackage{booktabs}%
\usepackage{algorithm}%
\usepackage{algorithmicx}%
\usepackage{algpseudocode}%
\usepackage{listings}%
\usepackage{lmodern}
\usepackage{subcaption}

\usepackage[T1]{fontenc}
\usepackage{mathdots}

\usepackage{adjustbox}
\usepackage{float}
\newtheorem*{theorem}{Theorem}
\newtheorem*{lemma}{Lemma}

\begin{document}

\title[Article Title]{Tight upper bound of the maximal quantum violation of Gisin's elegant Bell inequality and its application in randomness certification}


\author[1,2]{\fnm{Dan-Dan} \sur{Hu}}

\author[1,2]{\fnm{Meng-Yan} \sur{Li}}

\author*[1,2]{\fnm{Fen-Zhuo \sur{Guo}}}\email{gfenzhuo@bupt.edu.cn}

\author[3]{\fnm{Yu-Kun} \sur{Wang}}

\author[4]{\fnm{Hai-Feng} \sur{Dong}}

\author[5]{\fnm{Fei} \sur{Gao}}

\affil[1]{\orgdiv{School of Science}, \orgname{Beijing University of Posts and Telecommunications}, \city{Beijing}, \postcode{100876}, \country{China}}

\affil[2]{\orgdiv{Henan Key Laboratory of Network Cryptography Technology}, \city{Zhengzhou}, \postcode{450001}, \country{China}}

\affil[3]{\orgdiv{Beijing Key Laboratory of Petroleum Data Mining}, \orgname{China University of Petroleum-Beijing},  \city{Beijing}, \postcode{102249}, \country{China}}

\affil[4]{\orgdiv{School of Instrumentation Science and Opto-Electronics Engineering}, \orgname{Beihang University},  \city{Beijing}, \postcode{100191}, \country{China}}

\affil[5]{\orgdiv{State Key Laboratory of Networking and Switching Technology}, \orgname{Beijing University of Posts and Telecommunications},  \city{Beijing}, \postcode{100876}, \country{China}}


\abstract{The violation of a Bell inequality implies the existence of nonlocality, making device-independent randomness certification possible. This paper derives a tight upper bound for the maximal quantum violation of Gisin's elegant Bell inequality (EBI) for arbitrary two-qubit states, along with the constraints required to achieve this bound. This method provides the necessary and sufficient conditions for violating the EBI for several quantum states, including pure two-qubit states and the Werner states. The lower bound of certifiable global randomness is analyzed based on the tight upper bound of the EBI for pure two-qubit states, with a comparison to the Clauser-Horne-Shimony-Holt (CHSH) inequality. The relationship between the noise level and the lower bound of certifiable global randomness with respect to the Werner states is also explored, and the comparisons with both the CHSH inequality and the chained inequality are given. The results indicate that when the state approaches a maximally entangled state within specific quantified ranges, the EBI demonstrates advantages over both the CHSH inequality and the chained inequality, potentially enhancing practical device-independent randomness generation rates.}

\keywords{Bell nonlocality, Gisin's elegant Bell inequality, device independent, quantum randomness}



\maketitle

\section{Introduction}\label{sec1}

The Bell inequality, which is satisfied by local hidden variable (LHV) theories but violated by quantum mechanics, serves as a pivotal tool for distinguishing classical and quantum phenomena\cite{bell1964einstein}. With the advancement of quantum information science, a variety of Bell-type inequalities have arisen, including well-known examples such as the Clauser-Horne-Shimony-Holt (CHSH) inequality  \cite{clauser1969proposed}, the tilted CHSH inequality\cite{acin2012randomness}, Gisin's elegant Bell inequality (EBI)\cite{myrvold2009bell}, the chained inequality\cite{braunstein1990wringing} and the network Bell inequality\cite{tavakoli2022bell}. A violation of Bell inequality suggests the existence of quantum nonlocality, a fundamental resource for device-independent information processing tasks such as quantum key distribution\cite{brunner2014bell,arnon2021upper,su2023monte,lukanowski2023upper} and quantunm random number generation\cite{pironio2013security,pironio2010random,acin2016optimal,andersson2018device,woodhead2020maximal,xiao2023device,liu2024quantifying}.

Most Bell inequalities saturate their maximal quantum violation with respect to maximally entangled states\cite{clauser1969proposed,acin2012randomness,myrvold2009bell,braunstein1990wringing,svetlichny1987distinguishing,mermin1990extreme,salavrakos2017bell,augusiak2019bell}. However, in the experimental realizations, the challenge lies in the interference of noise, making the preparation of maximally entangled states a non-trivial task\cite{aolita2015open,piccolini2024asymptotically}. Consequently, the tight upper bound of maximal quantum violation of Bell inequalities for arbitrary quantum states and its application have investigated. In the bipartite scenario, Acín et al.\cite{acin2012randomness} and  Xiao et al.\cite{xiao2023device} derived the maximal quantum violations of the tilted CHSH inequality for pure two-qubit states and the chained inequality for arbitrary two-qubit quantum states, respectively. Additionally, they quantified the extractable randomness based on these maximal quantum violations. In the multipartite scenario, Li et al.\cite{li2017tight} established the maximal quantum violation of Svetlichny inequality for arbitrary three-qubit quantum states. Furthermore, the maximal quantum violation of Mermin inequality for arbitrary three-qubit quantum states was demonstrated and the result was extended to four-qubit MABK inequality\cite{siddiqui2019tight}.

Various quantum information processing scenarios may employ different Bell inequalities. The EBI is another inequality that can reveal bipartite quantum nonlocality, characterized by its elegant geometric structure in measurements\cite{andersson2017self}.  Due to the relatively large number of measurement settings in the EBI, it may offer advantages in certain quantum tasks, such as randomness certification. In this paper, we obtain the maximal quantum violation of the EBI for arbitrary two-qubit quantum states and explore its application in randomness certification. We derive a tight upper bound on the maximal quantum violation of the EBI using the singular values of the state correlation matrix. Constraints and the optimal measurement strategy required to achieve this bound are then presented. Additionally, we apply our results to pure two-qubit states and the Werner states to demonstrate the validity of our approach, and provide lower bounds for certifiable global randomness using semidefinite programs (SDP)\cite{navascues2007bounding,navascues2008convergent} based on the tight upper bounds. The lower bound of the global randomness relative to parameter of pure two-qubit states is discussed, and a comparison is made with the CHSH inequality. Moreover, we explore the relationship between noise level of the Werner states and the lower bound of global randomness, contrasting it with both the CHSH inequality and chained inequality.

\section{Tight upper bound of the maximal quantum violation of EBI}\label{sec2}

The EBI  involvs two parties, Alice and Bob, as introduced by Gisin et al.\cite{myrvold2009bell}. Alice has three dichotomic measurement settings $A_k(k=1,2,3)$ with corresponding outcome denoted as $a\in\{0,1\}$, Bob has four dichotomic measurement settings $B_l(l=1,2,3,4)$ with corresponding outcome denoted as $b\in\{0,1\}$. The Gisin's elegant Bell operator reads
\begin{equation} \label{eq1}
	S=A_1(B_1+B_2-B_3-B_4)+A_2(B_1-B_2+B_3-B_4)+A_3(B_1-B_2-B_3+B_4),
\end{equation}
where $A_k=\boldsymbol{a_k}\cdot\boldsymbol{\sigma}=\sum_{i=1}^{3}{a_k}_i\sigma_i$,  $B_l=\boldsymbol{b_l}\cdot\boldsymbol{\sigma}=\sum_{j=1}^{3}{b_l}_j\sigma_j$, $\boldsymbol{\sigma}=(\sigma_1,\sigma_2,\sigma_3)$, $\sigma_i(i=1,2,3)$ is the Pauli matrix, $\boldsymbol{a_k}=(a_{k1},a_{k2},a_{k3})$ and $\boldsymbol{b_l}=(b_{l1},b_{l2},b_{l3})$ are three dimensional real unit vectors. The classical bound of EBI for LHV theories is 6. It has been shown in Ref.\cite{acin2016optimal} that the maximal quantum violation of the EBI amounts to $4\sqrt{3}$, and it is realized with the two-qubit maximally entangled state 
\begin{equation}
		\left | \phi_+  \right \rangle =\frac{1}{\sqrt{2} }(\left | 00 \right \rangle +\left | 11  \right \rangle ) ,
\end{equation}
and the following measurements
\begin{equation}
		A_1=\sigma_1, A_2=\sigma_2, A_3=\sigma_3, \label{3}
\end{equation}
\begin{subequations}
\begin{equation}
			B_1=\frac{1}{\sqrt{3}}(\sigma_1-\sigma_2+\sigma_3), B_3=\frac{-1}{\sqrt{3}}(\sigma_1+\sigma_2+\sigma_3),
\end{equation}
\begin{equation}
			B_2=\frac{1}{\sqrt{3}}(\sigma_1+\sigma_2-\sigma_3), B_4=\frac{1}{\sqrt{3}}(-\sigma_1+\sigma_2+\sigma_3).
\end{equation}   \label{4}
\end{subequations}

The arbitrary two-qubit state $\rho$ in the Hilbert space $H=C_2\otimes C_2$ can be represented as
\begin{equation}
		\rho=\frac{1}{4}[I\otimes I+\sum_{i=1}^{3}r_i\sigma_i\otimes I+I\otimes\sum_{j=1}^{3}s_j\sigma_j+\sum_{i,j=1}^{3}{t_i}_j\sigma_i\otimes\sigma_j],
\end{equation}
where $T=(t_{ij})$ with $t_{ij}=tr(\sigma_i\otimes\sigma_j \rho)$ is the correlation matrix of the state $\rho$.
	
We derive the tight upper bound on the maximal quantum violation of EBI for arbitrary two-qubit state $\rho$ as follows.
\begin{theorem}
For an arbitrary two-qubit state $\rho$, the maximal quantum violation Q(S) of the EBI satisfies
\begin{equation}
			Q(S)_\rho \equiv \max_{A_k,B_l}\left |\left \langle S  \right \rangle _\rho  \right |\le 4\sqrt{\lambda _1^2+\lambda _2^2+\lambda _3^2} , \label{6}
\end{equation}
where $\left \langle S  \right \rangle _\rho=tr(\rho S)$, $\lambda _1,\lambda _2,\lambda _3$ are the three largest singular values of the matrix T.
\end{theorem}
	
Before proving the theorem, we introduce the following lemma from Ref. \cite{xiao2023device}.
\begin{lemma}
Let A be an $n \times n$ matrix and $A=U^\mathbb{T}\Sigma V$ the singuar value decomposition of $A$. Suppose the column vectors of the matrices $U$ and $V$ are $u_1,...,u_n$ and $v_1,...,v_n$, respectively. For any vectors $x=\sum x_iu_i$ and $y=\sum y_iv_i$, we have   
\begin{equation}
			\left |\boldsymbol{x}^\mathsf{T} A \boldsymbol{y}  \right | \le \sigma _{\max}	\left \|\boldsymbol{x} \right \| \left \|\boldsymbol{y} \right \| \label{lemma},
\end{equation}
where $\sigma _{\max} $ is the largest singular value of the matrix A and $\left \| \cdot \right \| $ stands for the 2-norm. The equality holds if and only if $\boldsymbol{u_i}$ and $\boldsymbol{v_i}$ are the left and right singular vectors corresponding to the maximum singular value, $x_iy_i$ keep the same symbol, and $x_i$ and $y_i$ are proportional to $i=1,...,n$.
\end{lemma}

\begin{proof}[Proof of the theorem]
Let $\left \langle A_k B_l \right \rangle $ denote the expectation value of the product of outcomes of Alice’s $k$th and Bob’s $l$th settings, then
\begin{align}
			\left \langle A_k B_l \right \rangle _\rho &= tr\left [\rho \left (\boldsymbol{a_k}\cdot\boldsymbol{\sigma}\otimes\boldsymbol{b}_l\cdot\boldsymbol{\sigma}  \right )  \right ] =\sum_{i,j=1}^{3}a_{ki}b_{lj}tr\left (\rho\sigma_i\otimes\sigma_j  \right ) \nonumber\\ &=\sum_{i,j=1}^{3}a_{ki}b_{lj}t_{ij}=\boldsymbol{a}_k^\mathsf{T}T\boldsymbol{b}_l.
\end{align}
The expectation value of Gisin's elegant Bell operator is thus given by
\begin{align}
			\left \langle S  \right \rangle _\rho &=\boldsymbol{a}_1^\mathsf{T}T\left (\boldsymbol{b}_1+\boldsymbol{b}_2-\boldsymbol{b}_3-\boldsymbol{b}_4  \right ) +\boldsymbol{a}_2^\mathsf{T}T\left (\boldsymbol{b}_1-\boldsymbol{b}_2+\boldsymbol{b}_3-\boldsymbol{b}_4  \right )+\boldsymbol{a}_3^\mathsf{T}T\left (\boldsymbol{b}_1-\boldsymbol{b}_2-\boldsymbol{b}_3+\boldsymbol{b}_4  \right ) \nonumber\\
			&\le \lambda_1\left |\boldsymbol{b}_1+\boldsymbol{b}_2-\boldsymbol{b}_3-\boldsymbol{b}_4  \right | +\lambda_2\left |\boldsymbol{b}_1-\boldsymbol{b}_2+\boldsymbol{b}_3-\boldsymbol{b}_4  \right | + \lambda_3\left |\boldsymbol{b}_1-\boldsymbol{b}_2-\boldsymbol{b}_3+\boldsymbol{b}_4  \right | \nonumber\\
			&\le\sqrt{\lambda_1^2+\lambda_2^2+\lambda_3^2}\nonumber\\
			&\quad\cdot\sqrt{\left |\boldsymbol{b}_1+\boldsymbol{b}_2-\boldsymbol{b}_3-\boldsymbol{b}_4  \right |^2+\left |\boldsymbol{b}_1-\boldsymbol{b}_2+\boldsymbol{b}_3-\boldsymbol{b}_4  \right |^2+\left |\boldsymbol{b}_1-\boldsymbol{b}_2-\boldsymbol{b}_3+\boldsymbol{b}_4  \right |^2} \nonumber\\
			&=\sqrt{\lambda_1^2+\lambda_2^2+\lambda_3^2}  \sqrt{12-2\left (\left \langle \boldsymbol{b_1},\boldsymbol{b_2} \right \rangle + \left \langle \boldsymbol{b_1},\boldsymbol{b_3} \right \rangle +\cdots +\left \langle \boldsymbol{b_3},\boldsymbol{b_4} \right \rangle \right ) }, \label{9}
		\end{align}
the first inequality comes from Lemma, and the second inequality is due to the Cauchy-Schwarz inequality.
		
To estimate the maximum value of \eqref{9}, we consider the following optimization problem:
	\begin{equation}
			\min\limits_{\boldsymbol{b}_i\in \mathbb{R}^3} \left ( \left \langle \boldsymbol{b_1},\boldsymbol{b_2} \right \rangle+\left \langle \boldsymbol{b_1},\boldsymbol{b_3} \right \rangle +\left \langle \boldsymbol{b_1},\boldsymbol{b_4} \right \rangle + \left \langle \boldsymbol{b_2},\boldsymbol{b_3} \right \rangle +\left \langle \boldsymbol{b_2},\boldsymbol{b_4} \right \rangle +\left \langle \boldsymbol{b_3},\boldsymbol{b_4} \right \rangle  \right ) . \label{10}
	\end{equation} 
		
Let $M=[m_{ij}]$ be the Gram matrix of the vectors $\left \{\boldsymbol{b_1},\boldsymbol{b_2},\boldsymbol{b_3},\boldsymbol{b_4}\right\} \in \mathbb{R}^3 $ with respect to the inner product
		
		\begin{equation} 
			M=\begin{pmatrix}
				\left \langle \boldsymbol{b_1},\boldsymbol{b_1} \right \rangle& \left \langle \boldsymbol{b_1},\boldsymbol{b_2} \right \rangle &  \left \langle \boldsymbol{b_1},\boldsymbol{b_3} \right \rangle & \left \langle \boldsymbol{b_1},\boldsymbol{b_4} \right \rangle\\
				\left \langle \boldsymbol{b_2},\boldsymbol{b_1} \right \rangle&  \left \langle \boldsymbol{b_2},\boldsymbol{b_2} \right \rangle &  \left \langle \boldsymbol{b_2},\boldsymbol{b_3} \right \rangle & \left \langle \boldsymbol{b_2},\boldsymbol{b_4} \right \rangle\\
				\left \langle \boldsymbol{b_3},\boldsymbol{b_1} \right \rangle&  \left \langle \boldsymbol{b_3},\boldsymbol{b_2} \right \rangle & \left \langle \boldsymbol{b_3},\boldsymbol{b_3} \right \rangle & \left \langle \boldsymbol{b_3},\boldsymbol{b_4} \right \rangle\\
				\left \langle \boldsymbol{b_4},\boldsymbol{b_1} \right \rangle&\left \langle \boldsymbol{b_4},\boldsymbol{b_2} \right \rangle  & \left \langle \boldsymbol{b_4},\boldsymbol{b_3} \right \rangle &\left \langle \boldsymbol{b_4},\boldsymbol{b_4} \right \rangle
			\end{pmatrix},
		\end{equation}
where all $\boldsymbol{b}_i$ are unit vectors and $M\succeq 0$. Define
		\begin{equation}
			W=\begin{pmatrix}
				0 &1  & 1 & 1\\
				1&  0& 1 &1 \\
				1& 1 & 0 &1 \\
				1& 1 &1 &0
			\end{pmatrix}.
		\end{equation}
We reformulate the optimization problem \eqref{10} into the following SDP:
		\begin{equation}
			\begin{aligned}
				&\min_{M}  \quad \frac{1}{2}tr\left ( MW \right ), \\
				&s.t. \quad M\succeq 0,m_{ii}=1,\forall i.
			\end{aligned}
		\end{equation}
The dual problem of the SDP can be derived as follows:
		\begin{equation}
			\begin{aligned}
				&\max_{v}  \quad tr(diag(v)), \\
				&s.t. \quad \frac{1}{2}W-diag(v) \succeq 0.
			\end{aligned}
		\end{equation}
		
We choose unit vectors $\boldsymbol{b}_i$ to be of the form
		\begin{subequations}
			\begin{equation}
				\boldsymbol{b_1}=\frac{1}{\sqrt{\lambda_1^2+\lambda_2^2+\lambda_3^2}}\left (\lambda_1, -\lambda_2, \lambda_3 \right ) ^\mathsf{T},
			\end{equation}
			\begin{equation}
				\boldsymbol{b_2}=\frac{1}{\sqrt{\lambda_1^2+\lambda_2^2+\lambda_3^2}}\left (\lambda_1, \lambda_2, -\lambda_3 \right ) ^\mathsf{T},
			\end{equation}
			\begin{equation}
				\boldsymbol{b_3}=\frac{-1}{\sqrt{\lambda_1^2+\lambda_2^2+\lambda_3^2}}\left (\lambda_1, \lambda_2, \lambda_3 \right ) ^\mathsf{T},
			\end{equation}
			\begin{equation}
				\boldsymbol{b_4}=\frac{1}{\sqrt{\lambda_1^2+\lambda_2^2+\lambda_3^2}}\left (-\lambda_1, \lambda_2, \lambda_3 \right ) ^\mathsf{T},
			\end{equation}
		\end{subequations}
where $\lambda _1,\lambda _2,\lambda _3$ are the three largest singular values of the matrix T. We can see that
		\begin{subequations} \label{16}
			\begin{equation}
				\left \langle \boldsymbol{b_1},\boldsymbol{b_2}  \right \rangle =\left \langle \boldsymbol{b_3},\boldsymbol{b_4}  \right \rangle  =\frac{1}{\sqrt{\lambda_1^2+\lambda_2^2+\lambda_3^2}}\left ( \lambda_1^2-\lambda_2^2- \lambda_3^2  \right ) ,
			\end{equation}
			\begin{equation}
				\left \langle \boldsymbol{b_1},\boldsymbol{b_3}  \right \rangle =\left \langle \boldsymbol{b_2},\boldsymbol{b_4}  \right \rangle  =\frac{1}{\sqrt{\lambda_1^2+\lambda_2^2+\lambda_3^2}}\left (-\lambda_1^2+\lambda_2^2- \lambda_3^2  \right ),
			\end{equation}
			\begin{equation}
				\left \langle \boldsymbol{b_1},\boldsymbol{b_4}  \right \rangle =\left \langle \boldsymbol{b_2},\boldsymbol{b_3}  \right \rangle  =\frac{1}{\sqrt{\lambda_1^2+\lambda_2^2+\lambda_3^2}}\left (-\lambda_1^2-\lambda_2^2+ \lambda_3^2  \right ).
			\end{equation}
		\end{subequations}
Therefore, the value of the primal problem is given by
		\begin{equation}
			p'=-2,
		\end{equation}
the corresponding Gram matrix $M'$, whose elements are given by equation \eqref{16}, and it evidently satisfies the constraint $M'\succeq 0,m_{ii}'=1,\forall i$.
		
Next, taking the 4-dimensional vector
		\begin{equation}
			v'=- \frac{1}{2}\left (1,1,1,1  \right ) , \label{18}
		\end{equation}
we find that the matrix $W$ has a minimum eigenvalue of $-1$, and the matrix $diag(v')$ has a maximum eigenvalue of $-\frac{1}{2}$. Noting that $W$ and $diag(v')$ are both Hermitian, we have
		\begin{align}
			\lambda_{\min}\left ( \frac{1}{2}W-diag(v') \right ) & \ge	\lambda_{\min}\left ( \frac{1}{2}W\right )-\lambda_{\max}\left ( diag(v') \right )  \nonumber \\
			&=-\frac{1}{2} - \left (-\frac{1}{2}  \right ) =0, \label{19}
		\end{align}
where $\lambda_{min}/\lambda_{max}(M)$ is the minimum/maximum eigenvalue of matrix $M$. Thus, the $diag(v')$ defined by \eqref{18} ensures that the constraint $\frac{1}{2}W - diag(v') \succeq 0 $ is satisfied. The value of the dual problem is
		\begin{equation}
			d'=-2.
		\end{equation}
		
Since $d'=p'$, strong duality holds, implying that $M'$ and $v'$ are the optimal solutions of the primal and dual problems, respectively. Finally, from \eqref{9}, we have $\left \langle S  \right \rangle _\rho \le 4\sqrt{\lambda_1^2+\lambda_2^2+\lambda_3^2}$, proving the Theorem.
	\end{proof}
  
\section*{Tightness of the upper bound}
    
Let us now discuss the conditions where upper bound \eqref{6} saturates.  From the lemma, the first inequality in \eqref{9} is saturated if $\boldsymbol{b_1}+\boldsymbol{b_2}-\boldsymbol{b_3}-\boldsymbol{b_4}$,  $\boldsymbol{b_1}-\boldsymbol{b_2}+\boldsymbol{b_3}-\boldsymbol{b_4}$, $\boldsymbol{b_1}-\boldsymbol{b_2}-\boldsymbol{b_3}+\boldsymbol{b_4}$ are the singular vectors corresponding to the three largest singular values $\lambda_1$, $\lambda_2$, $\lambda_3$ and $\boldsymbol{a_1}=\frac{T\left (\boldsymbol{b_1}+\boldsymbol{b_2}-\boldsymbol{b_3}-\boldsymbol{b_4}\right ) }{\left | T\left (\boldsymbol{b_1}+\boldsymbol{b_2}-\boldsymbol{b_3}-\boldsymbol{b_4}\right ) \right | }$, $\boldsymbol{a_2}=\frac{T\left (\boldsymbol{b_1}-\boldsymbol{b_2}+\boldsymbol{b_3}-\boldsymbol{b_4}\right ) }{\left | T\left (\boldsymbol{b_1}-\boldsymbol{b_2}+\boldsymbol{b_3}-\boldsymbol{b_4}\right ) \right | }$, $\boldsymbol{a_3}=\frac{T\left (\boldsymbol{b_1}-\boldsymbol{b_2}-\boldsymbol{b_3}+\boldsymbol{b_4}\right ) }{\left | T\left (\boldsymbol{b_1}-\boldsymbol{b_2}-\boldsymbol{b_3}+\boldsymbol{b_4}\right ) \right | }$. According to the Cauchy-Schwartz inequality, the second inequality in \eqref{9} is saturated if $\frac{\lambda_1}{\left |  \boldsymbol{b_1}+\boldsymbol{b_2}-\boldsymbol{b_3}-\boldsymbol{b_4} \right | } = \frac{\lambda_2}{\left |  \boldsymbol{b_1}-\boldsymbol{b_2}+\boldsymbol{b_3}-\boldsymbol{b_4} \right | } = \frac{\lambda_3}{\left |  \boldsymbol{b_1}-\boldsymbol{b_2}-\boldsymbol{b_3}+\boldsymbol{b_4} \right | }$. Moreover, to obtain the maximum value, we require the condition $\left \langle \boldsymbol{b_1},\boldsymbol{b_2} \right \rangle+\left \langle \boldsymbol{b_1},\boldsymbol{b_3} \right \rangle +\left \langle \boldsymbol{b_1},\boldsymbol{b_4} \right \rangle + \left \langle \boldsymbol{b_2},\boldsymbol{b_3} \right \rangle +\left \langle \boldsymbol{b_2},\boldsymbol{b_4} \right \rangle +\left \langle \boldsymbol{b_3},\boldsymbol{b_4} \right \rangle  =-2$.
	
As discussed above, we can conclude that the tight upper bound of the maximal quantum violation is achieved when Alice's and Bob's measurements exhibit specific geometric configurations in the Bloch ball.  The six eigenstates of Alice’s projective measurements form a complete set of three mutually unbiased bases, with each basis corresponding to a pair of opposite corners of an octahedron in the Bloch sphere, as shown in Fig \ref{1}. The eight eigenstates of Bob’s projective measurements can be grouped into two positive operator-valued measures (POVMs): one formed by the $+1$ outcome projectors and the other by the $-1$ outcome projectors. The $+1$ and $-1$ outcomes of each projective measurement correspond to a pair of opposite corners of a cuboid in the Bloch sphere, as shown in Fig \ref{2}.
    \begin{figure*}[t]
    \centering
    \begin{subfigure}[b]{0.48\linewidth}
        \centering
        \includegraphics[width=\linewidth]{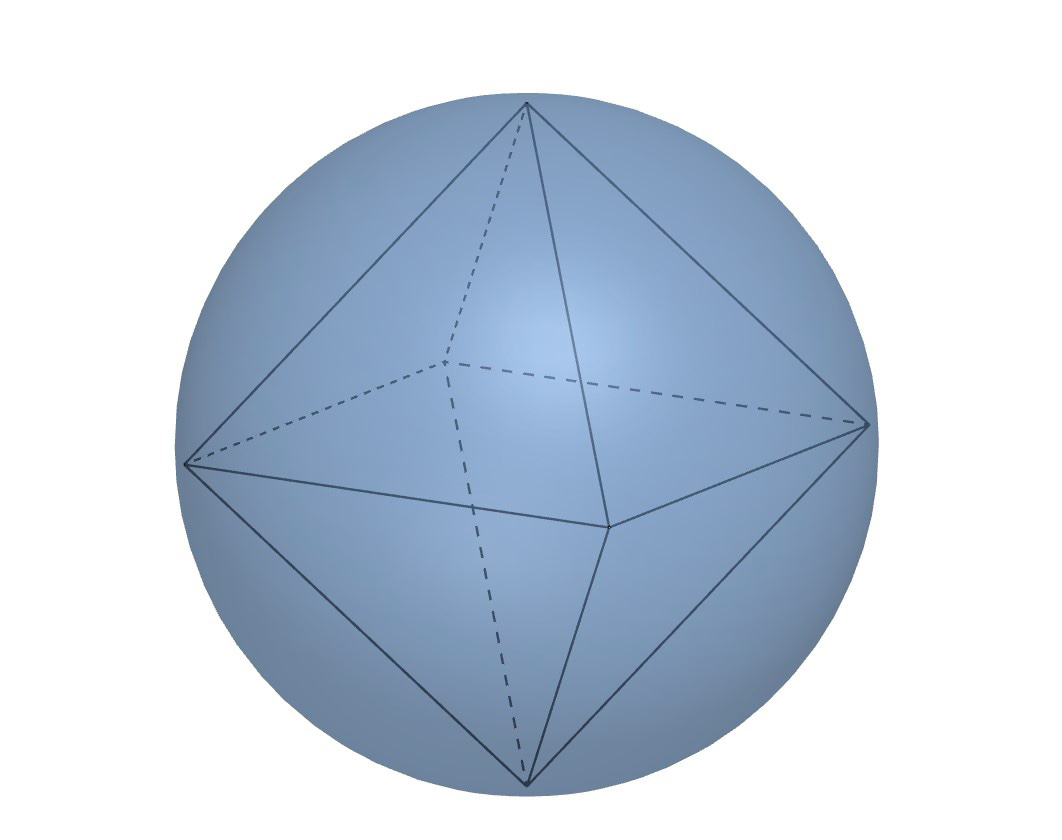}
        \caption{The octahedron in Alice's Bloch sphere.}
        \label{1}
    \end{subfigure}
    \hfill
    \begin{subfigure}[b]{0.48\linewidth}
        \centering
        \includegraphics[width=\linewidth]{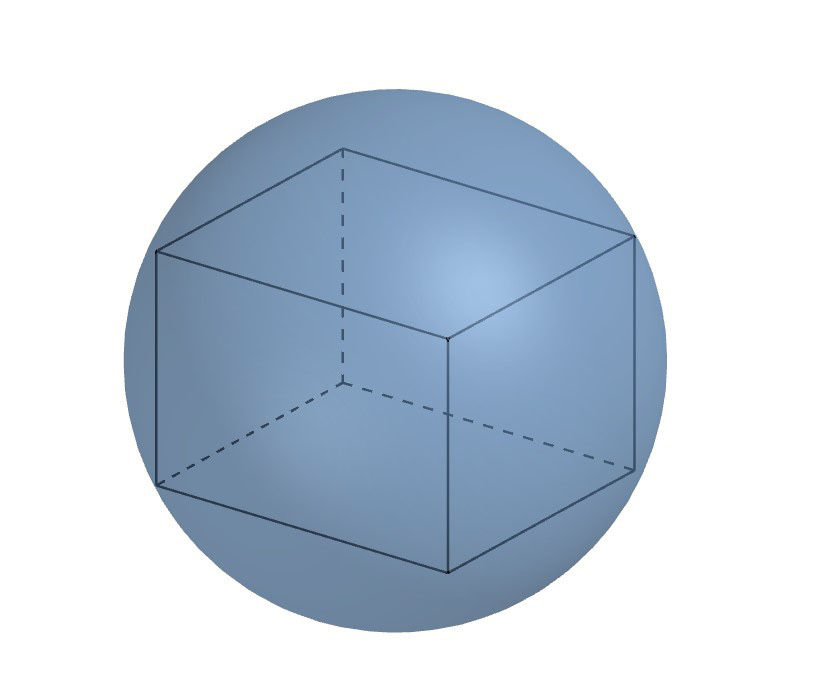}
        \caption{The cuboid in Bob's Bloch sphere.}
        \label{2}
    \end{subfigure}
    \caption{(a) Alice’s measurement eigenstates are positioned at the vertices of an octahedron. (b) Bob’s measurement eigenstates are positioned at the vertices of a cuboid.}
    
\end{figure*}

\textit{Example 1.}  
 Consider the singlet state
    \begin{equation}
    	\left |\psi^-  \right \rangle =\frac{1}{\sqrt{2}}( \left | 01 \right \rangle -\left | 10 \right \rangle), 
    \end{equation}
the correlation matrix of $\rho= \lvert \psi^-  \rangle  \langle \psi^- \lvert  $ is
    \begin{equation}
    	T=\begin{pmatrix}
    		-1& & \\
    		& -1& \\
    		& &-1
    	\end{pmatrix},
    \end{equation}
the singular values of the matrix $T$ are $1, 1,1$, thus
    \begin{equation}
    	Q\left (S  \right )_{\rho}\le 4\sqrt{3}. 
    \end{equation}
The maximal quantum violation of \(4\sqrt{3}\) can be achieved when Alice’s measurements are given by \eqref{3} and Bob’s measurements are given by \eqref{4}. \eqref{3} and \eqref{4} satisfy the conditions for the tightness of the upper bound (the measurement strategy in the following examples also meets the tightness conditions), demonstrating that the upper bound in \eqref{6} is tight.
    
\textit{Example 2.}  
Consider a pure two-qubit state, which can always be written as
    \begin{equation}
    	\left |\psi_\theta  \right \rangle =\cos{\theta}\left | 00 \right \rangle +\sin{\theta}\left | 11 \right \rangle,  \label{21}
    \end{equation}
where $\theta \in\left [0,\frac{\pi}{4}  \right ] $, the correlation matrix of $\rho\left (\theta  \right ) = \lvert \psi_\theta  \rangle  \langle \psi_\theta \lvert  $ is
    \begin{equation}
    	T=\begin{pmatrix}
    		\sin 2\theta& & \\
    		& \sin 2\theta& \\
    		& &1
    	\end{pmatrix},
    \end{equation}
the singular values of the matrix $T$ are $\sin 2\theta, \sin 2\theta,1$, hence
    \begin{equation}
    	Q\left (S  \right )_{\rho\left (\theta  \right ) }\le 4\sqrt{1+2\sin^{2} 2\theta}. \label{23}
    \end{equation}
The maximal quantum violation $4\sqrt{1+2\sin^{2} 2\theta}$ is achieved when Alice and Bob use the measurements given by \eqref{27} and \eqref{28}, confirming that the upper bound also tightly applies to the pure two-qubit states.
    \begin{equation}
    	A_1=\sigma_1, A_2=\sigma_2, A_3=\sigma_3, \label{27}
    \end{equation}
    \begin{subequations}
    	\begin{equation}
    		B_1=\frac{1}{\sqrt{1+2\sin^22\theta}}\left (\sin2\theta\sigma_1-\sin2\theta\sigma_2+\sigma_3 \right ) ,
    	\end{equation}
    	\begin{equation}
    		B_2=\frac{1}{\sqrt{1+2\sin^22\theta}}\left (\sin2\theta\sigma_1+\sin2\theta\sigma_2-\sigma_3  \right ) ,
    	\end{equation}
    	\begin{equation}
    		B_3=\frac{1}{\sqrt{1+2\sin^22\theta}}\left (-\sin2\theta\sigma_1-\sin2\theta\sigma_2-\sigma_3  \right ) ,
    	\end{equation}
    	\begin{equation}
    		B_4=\frac{1}{\sqrt{1+2\sin^22\theta}}\left (-\sin2\theta\sigma_1+\sin2\theta\sigma_2+\sigma_3  \right ) .
    	\end{equation} \label{28}
    \end{subequations} 
    
\textit{Example 3.}
Consider the Werner state
    \begin{equation}
    	\rho\left (p  \right )=p \lvert \phi^+  \rangle \langle \phi^+\lvert +\left (1-p  \right ) \frac{I}{4},
    \end{equation}
where $0 \le\ p \le 1$, and $I$ is the identity matrix on $\mathbb{C}^2 \otimes \mathbb{C}^2$. The correlation matrix of $\rho\left (p  \right )$ is
    \begin{equation}
    	T=\begin{pmatrix}
    		p& & \\
    		& p& \\
    		& &p
    	\end{pmatrix},
    \end{equation}
the singular values of the matrix $T$ are $p, p, p$. Therefore 
    \begin{equation}
    	Q\left (S  \right )_{\rho\left (p  \right )} \le4\sqrt{3}p.
    \end{equation}
The maximal quantum violation $4\sqrt{3}p$ is achieved when Alice's measurements are defined by \eqref{3} and Bob's measurements by \eqref{4}, which shows that the upper bound also tightly applies to the Werner states. 
    
\section{Randomness certification based on EBI}
    
The violation of Bell inequality implies the existence of inherent randomness \cite{herrero2017quantum}. In this section, we establish a lower bound on global randomness based on the maximal quantum violation of the EBI. This randomness is device-independent, as we do not make any assumptions about the sources and measurements of the device.
    
In the quantum region, for a given quantum distribution $P=\left \{ p(ab|xy) \right \}$, the randomness of output pairs $(a,b)$ conditioned on input pairs $(x,y)$ represents global randomness, quantified in terms of min-entropy
    \begin{equation}
    	H_\infty \left (AB|XY\right ) =-\log_2\max_{ab} P\left ( ab|xy \right ) .
    \end{equation}
    
We aim to establish the lower bound on global randomness based on Bell inequality violation $I$ using SDP. In the device-independent scenario, where the internal device structure is not considered, maximizing $P\left ( ab|xy \right ) $ requires running over all quantum realizations compatible with $P$. However, this task is computationally challenging due to the complexity of considering all quantum realizations. Therefore, we employ numerical method via SDP to obtain the maximum value of $P\left ( ab|xy \right ) $ posed by the following optimization problem:
    \begin{equation}
    	P^{*}(ab|xy)=\max_{\rho,M_{x}^{a},M_{y}^{b} }P(ab|xy), 
    \end{equation}
    \begin{equation}
    	s.t. \quad \sum_{abcy}c_{abxy}P(ab|xy)=I,
    \end{equation}
    \begin{equation}
    	P(ab|xy)=tr(\rho M_{x}^{a}\otimes M_{y}^{b}).
    \end{equation}

We consider the pure two-qubit state $\left |\psi_\theta  \right \rangle =\cos{\theta}\left | 00 \right \rangle +\sin{\theta}\left | 11 \right \rangle$, where the parameter $\theta$ determines amount of entanglement  in the two-qubit state. The state is maximally entangled  at $\theta = \frac{\pi}{4}$. The maximal quantum violation of EBI is $4\sqrt{1+2\sin^{2} 2\theta}$. The CHSH inequality is violated for $\theta>0$, while the EBI is violated for $\theta>0.456$. We optimize the lower bound of the global randomness based on the quantum violation $I$ of Bell inequality on the Q2 level of the  Navascués-Pironio-Acín (NPA) hierarchy\cite{navascues2008convergent}. As shown in Fig \ref{fig2}, randomness can only be extracted when the Bell inequality is violated. The 1.34 bits of randomness can be certified from the maximal violation of the EBI. For all $\theta>0.67$, the lower bound of randomness certified by the EBI exceeds that of the CHSH inequality, indicating that EBI performs better in randomness generation rates. 
   \begin{figure}[H]
    	\centering
    	\includegraphics[width=200pt]{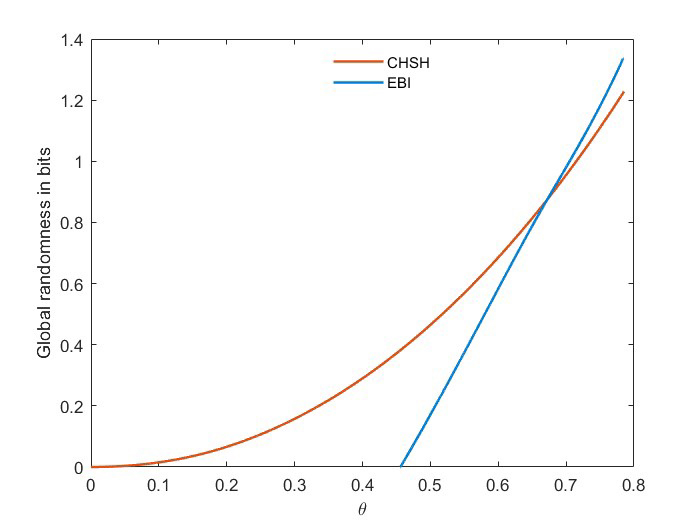}
    	\caption{Lower bounds on the randomness as a function of the parameter 
    		$\theta$ of the state. The orange curve is obtained by the CHSH inequality violation, while the blue curve is obtained by the EBI violation.}
    	\label{fig2}
    \end{figure}
Similarly, we consider the Werner state $\rho\left (p  \right )=p \lvert \phi^+  \rangle  \langle \phi^+ \lvert +\left (1-p  \right ) \frac{I}{4}$. As illustrated in Fig \ref{fig3}, if $p>0.94$, the extractable randomness using the EBI surpasses that of obtained from the chained inequality with three measurements on each side. Furthermore, for $p>0.965$, the extractable randomness using the EBI exceeds that of the CHSH inequality. The maximum global randomness certified by the CHSH inequality, chained inequality, and EBI is 1.23 bits, 1.1 bits, and 1.34 bits, respectively. Therefore, for a given $p$ approaching 1, the EBI surpasses both the CHSH inequality and the chained inequality in terms of randomness generation rates.
 \begin{figure}[H]
    	\centering
    	\includegraphics[width=200pt]{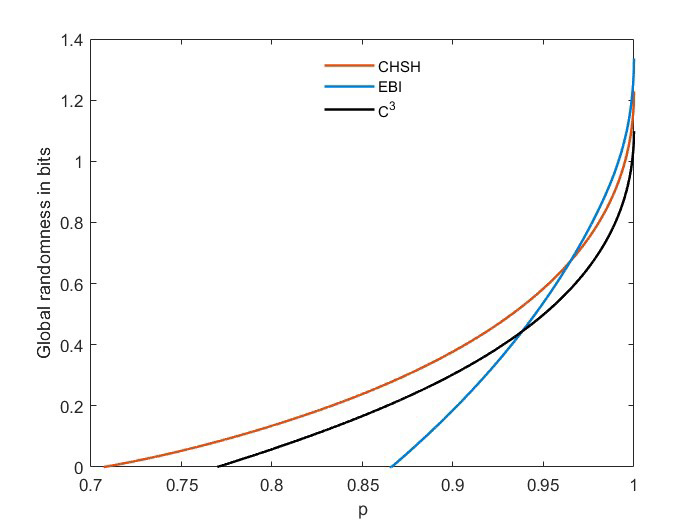}
    	\caption{Lower bounds on the randomness as a function of the noise level $p$  of the state. The curves are obtained by the CHSH inequality violation (orange), the EBI violation (blue), and the chained inequality violation with three measurements each party (black), respectively.}
    	\label{fig3}
    \end{figure}
    
\section{Conclusion and discussion}
Through singular value decomposition and the solving of SDP, we have deduced a tight upper bound on the maximal quantum violation of the EBI and identified the conditions to achieve this bound. By analyzing the optimal measurement strategy, we calculate the tight bound for pure two-qubit states and Werner states, and explore its application in device-independent randomness certification. By comparing the lower bounds on the global randomness certified by the EBI with those of the CHSH and chained inequalities, based on the tight upper bound, we find that the EBI has higher randomness generation rates than both CHSH and chained inequalities in certain regions. This result provides a theoretical foundation for the design of device-independent random number generators.
    
The global randomness certification in our work is based on the quantum violation of the EBI at the Q2 level of the NPA hierarchy. Currently, there is no analytical functional relationship between global randomness and the violation of Bell inequalities. We leave the research to deduce the relationship between global randomness and the violation of EBI for future work.

Note added: During the review process of this work (submitted to EPJ Quantum Technology on July 24, 2024), we became aware of Ref. \cite{bhowmick2024necessary}. While it presents some overlapping results, our proof technique, which is a central conceptual contribution of this work, differs substantially. Specifically, our approach leverages semidefinite programming and matrix singular value decomposition besides the Cauchy-Schwarz inequality. Additionally, we apply our findings to device-independent randomness certification, which constitutes another significant contribution of this work.

\makeatletter
\renewcommand\@biblabel[1]{#1.}
\makeatother
\small\bibliography{sn-articl}
\end{document}